\documentclass[a4paper,english]{article}
\usepackage[utf8]{inputenc}

\usepackage{amsmath}
\usepackage{amsfonts}
\usepackage{mathtools}   %
\usepackage{empheq}
\usepackage{dsfont}

\usepackage{comment}
\usepackage{enumerate}
\usepackage{paralist}

\usepackage{a4wide}
\usepackage[noend]{algpseudocode}

\usepackage{tabularx}
\usepackage{multirow}
\usepackage{booktabs}

\usepackage{authblk}
\usepackage{amsthm}
\newtheorem{theorem}{Theorem}{\bfseries}{\normalfont}
{\bfseries}{\normalfont}
\newtheorem{proposition}{Proposition}{\bfseries}{\normalfont}
\newtheorem{rrule}{Rule}{\bfseries}{\normalfont}
\newtheorem{corollary}{Corollary}{\bfseries}{\normalfont}
{\bfseries}{\normalfont}
\newtheorem{lemma}{Lemma}{\bfseries}{\normalfont}

\usepackage{tikz}
\usetikzlibrary{positioning,backgrounds,patterns,calc,shapes,arrows}

\usepackage{bbding}
\usepackage{placeins}

\def \eps {\varepsilon}

\DeclareMathOperator{\dist}{dist}

\DeclareMathOperator{\poly}{poly}
\DeclareMathOperator{\opt}{opt}
\DeclareMathOperator{\val}{val}

\usepackage{microtype,ellipsis}

\usepackage[pdfdisplaydoctitle,pagebackref,menucolor=orange!40!black,filecolor=magenta!40!black,urlcolor=blue!40!black,linkcolor=red!40!black,citecolor=green!40!black,colorlinks]{hyperref}

\hypersetup{pdftitle={Shortest Path Most Vital Edges}, pdfauthor={}}

\usepackage{algorithm}

\usepackage{units}

\usepackage[numbers,sort]{natbib}
\makeatletter
\def\NAT@spacechar{~}%
\makeatother

\usepackage[sort&compress,nameinlink]{cleveref} %

\theoremstyle{remark}
\newtheorem{remark}{Remark}

\crefname{rrule}{Rule}{Rules}
\crefname{line}{Line}{Lines}
\crefname{equation}{Equation}{Equations}
\crefname{theorem}{Theorem}{Theorems}
\crefname{lemma}{Lemma}{Lemmas}
\crefname{section}{Section}{Sections}
\crefname{table}{Table}{Tables}
\Crefname{obs}{Observation}{Observations}
\crefname{obs}{Observation}{Observations}
\crefname{figure}{Figure}{Figures}
\crefname{corollary}{Corollary}{Corollary}
\crefname{proposition}{Proposition}{Propositions}
\DeclareRobustCommand{\NoKernelAssume}{$\text{NP}\subseteq \text{{coNP/poly}}$}

\title{A More Fine-Grained Complexity Analysis of Finding the Most Vital Edges for Undirected Shortest Paths}

\author[1]{Cristina Bazgan\thanks{Institut Universitaire de France}$^,$}
\author[2]{Till~Fluschnik\thanks{Supported by DFG research project DAMM, NI~369/13.}$^,$}
\author[2]{Andr\'e Nichterlein}
\author[2]{Rolf Niedermeier}
\author[2]{Maximilian~Stahlberg}

\affil[1]{%
\small{
Universit\'e Paris-Dauphine, PSL Research University, CNRS, UMR 7243, LAMSADE, 75016 Paris, France\\
\texttt{bazgan@lamsade.dauphine.fr}
}}
\affil[2]{
\small{
Institut f\"ur Softwaretechnik und Theoretische Informatik, TU Berlin, Berlin, Germany\\
\texttt{\{till.fluschnik,andre.nichterlein,rolf.niedermeier,maximilian.stahlberg\}@tu-berlin.de}
}}

\usepackage{etoolbox}

\newcounter{Bew1}
\newcounter{Bew2}

\newcommand{\MVE}{\textsc{Shortest Path Most Vital Edges}\xspace}
\newcommand{\SPMVE}{{SP-MVE}\xspace}

\newcommand{\kmve}{\textsc{Max-Length SP-MVE}\xspace}
\newcommand{\meb}{\textsc{Min-Cost SP-MVE}\xspace}

\newcommand{\VC}{\textsc{Vertex Cover}\xspace}

\newcommand{\N}{\mathds{N}}
\newcommand{\I}{I}

\newcommand{\R}{\mathbb{R}}

\newcommand{\tw}{tw}

\crefname{enumi}{Step}{Steps}

\newcommand{\problemdef}[3]{
    \begin{center}
    \begin{minipage}{0.95\textwidth}
      \noindent
      \normalsize\textsc{#1}
      
      \vspace{1pt}
      \setlength{\tabcolsep}{3pt}
      \renewcommand{\arraystretch}{1.0}
      \begin{tabularx}{\textwidth}{@{}lX@{}}
	\normalsize\textbf{Input:} 	& \normalsize#2 \\
	\normalsize\textbf{Question:} 	& \normalsize#3
      \end{tabularx}
    \end{minipage}
    \end{center}
}

\date{}
\begin{document}

\maketitle
\begin{abstract}
	We study the NP-hard \MVE problem arising in the context of analyzing network robustness.
	For an undirected graph with positive integer edge lengths and two designated vertices~$s$ and~$t$, the goal is to delete as few edges as possible in order to increase the length of the (new) shortest $st$-path as much as possible.
	This scenario has been studied from the viewpoint of parameterized complexity and approximation algorithms.
	We contribute to this line of research by providing refined computational tractability as well as hardness results. 
	We achieve this by a systematic investigation of various problem-specific parameters and their influence on the computational complexity. 
	Charting the border between tractability and intractability, we also identify numerous challenges for future research.

\end{abstract}

\section{Introduction}

Finding shortest paths in graphs is arguably among the most fundamental graph
problems.
We study the case of undirected graphs with positive integer edge lengths within the framework of ``most vital edges'' or (equivalently) ``interdiction'' or ``edge blocker'' problems.
That is, we are interested in the scenario where the goal is to delete (few) edges such that in the resulting graph the shortest $st$-path gets (much) longer.
This is motivated by applications in investigating robustness and critical infrastructure in the context of network design.
Our results provide new insights with respect to classical, parameterized, and approximation complexity of this fundamental edge deletion problem which is known to be NP-hard~\cite{BKS95,KBBEGRZ08}. In its decision version, the problem 
reads as follows.

        \problemdef
                {\MVE (SP-MVE)}
                {An undirected graph~$G = (V,E)$ with positive edge lengths $\tau\colon E \rightarrow \N$, two vertices~$s,t \in V$, and integers~$k,\ell \in \N$.}
                {Is there an edge subset~$S \subseteq E$, $|S| \le k$, such that the length of a shortest $st$-path in $G-S$ is at least~$\ell$?}
We set~$b := \ell - \dist_G(s,t)$ to be the number by which the length of every shortest $st$-path shall be increased.
If all edges have length one, then we say that the graph has unit-length edges.
Naturally, SP-MVE comes along with two optimization versions: Either delete as few edges as possible in order to achieve a length increase of at least~$b$ (called \meb) or obtain a maximum length increase under the constraint that $k$~edges can be deleted (called \kmve).
For an instance of SP-MVE or \kmve we assume that $k$ is smaller than the size of any~$st$-edge-cut in the input graph. %
Otherwise, removing all edges of a minimum-cardinality $st$-edge-cut (which is polynomial-time computable) would lead to a solution disconnecting $s$ and~$t$.

\paragraph*{Related work.}
Due to the immediate practical relevance, e.g. in supply~\cite{FH77,golden1978problem} and communication~\cite{ItaiPS82} networks, there are numerous studies concerning ``most vital edges (and vertices)'' and related problems.
We focus on shortest paths, but there are further studies for problems such as \textsc{Minimum Spanning Tree}~\cite{BTV12,BTV13b,FS96,GS14,Lia01} or \textsc{Maximum Flow}~\cite{GS14,PS16,Woo93}, to mention only two.
With respect to shortest path computation, the following is known.

First, we mention in passing that a general
result of \citet{FH77} implies that allowing
the subdivision of edges instead of edge deletions as modification
operation makes the problem polynomial-time solvable.
Notably, it also has been studied to find \emph{one} most
vital edge of a shortest path; this can be solved in almost linear
time~\cite{NPW01}.

\citet{BKS95} showed that SP-MVE is NP-complete.
\citet{KBBEGRZ08} found polynomial-time constant-factor inapproximability results for both optimization versions.
For the case of directed graphs, \citet{IW02} provided heuristic solutions based on mixed-integer programming together with experimental results.
\citet{PS16} studied the restriction of the directed case to
planar graphs and again obtained NP-hardness results.

\citet{BaierEHKKPSS10} studied a minimization variant of SP-MVE 
where edges, in addition to a length value, also have a deletion cost associated with them. They 
refer to this problem as
\textsc{Minimum Length-Bounded Cut (MLBC)} and
showed that it is NP-hard to approximate within a factor of 1.1377 for~$\ell\geq 5$.
Moreover, they developed a polynomial-time algorithm for the 
special case of~$b=1$.
Further, they showed that MLBC with general edge-costs and edge-lengths remains NP-hard on series-parallel and outerplanar graphs.

\citet{GolovachT11} studied SP-MVE with unit-length edges under the name \textsc{Bounded Edge Undirected Cut (BEUC)} from a parameterized complexity point of view.
They showed that SP-MVE with unit-length edges is W[1]-hard with respect to~$k$ and that it is fixed-parameter tractable with respect to the combined parameter~$(k,\ell)$.
Answering an open question of \citet{GolovachT11}, \citet{FHNN18} showed that SP-MVE with unit-length edges does not admit a polynomial-size problem kernel with respect to~$(k,\ell)$, unless \NoKernelAssume.
Moreover, the latter showed that SP-MVE remains NP-hard on planar graphs.
\citet{DvorakK15} also studied SP-MVE with unit-length edges.
They showed that the problem is W[1]-hard with respect to pathwidth.
On the positive side, they showed that the problem is fixed-parameter tractable with respect to the treedepth of the input graph and with respect to~$\ell$ and the treewidth~$\tw$ of the input graph combined. 
Upon the latter, they proved that SP-MVE does not admit a polynomial-size problem kernel with respect to~$(\ell,\tw)$, unless \NoKernelAssume.
Moreover, they developed an algorithm running in~$n^{O(\tw^2)}$ time, 
that is, they showed that the problem lies in the complexity class~XP when parameterized by~$\tw$.
\citet{Kolman18} studied SP-MVE and its vertex deletion variant.
He proved that both variants on planar graphs are fixed-parameter tractable when parameterized by~$\ell$.
Additionally, for the vertex-deletion variant, he developed an~$O(\tw\cdot\sqrt{\log\tw})$-approximation algorithm, which improves to a~$\tw$-approximation algorithm when the tree decomposition is given.

\paragraph{Our results.}
We perform an extensive study of multivariate
complexity aspects~\cite{FJR13,Nie10} of SP-MVE. More specifically,
we perform a refined complexity analysis in terms of
how certain problem-specific parameters influence the computational
complexity of SP-MVE and its optimization variants.
The parameters we study include aspects of graph structure as well as
special restrictions on the problem parameters.
We also report a few findings on (parameterized) approximability.
Let us feature three main conclusions from our work:
First,  it is known that harming the network only a little bit (that is, $b=1$) is doable in polynomial time~\cite{BaierEHKKPSS10} while
we show that harming the network slightly more (that is, $b\geq 2$) 
becomes NP-hard. 
Second, the ``cluster vertex deletion number'', advocated by \citet{DK12} as a parameterization between vertex cover number and clique-width, currently is our most interesting parameter that yields fixed-parameter tractability for \SPMVE with unit-length edges.
Third, with general edge-lengths \SPMVE remains NP-hard even on complete graphs.
\cref{fig:par-hier} surveys our current understanding of the parameterized complexity of SP-MVE with respect to a number of well-known graph parameters, identifying numerous open questions.
\begin{figure}[t!]
	\centering
	\newcommand{\tworowsSec}[2]{#1 #2}
	\newcommand{\disttoSec}[1]{distance~to #1}
	\tikzstyle{boxes}=[draw,thick, rounded corners=3mm,text width=2.6cm,align=center,text opacity=1,fill opacity=1,fill=white]
	\tikzstyle{unk}=[fill=gray!15!white]
	\def\eps{0.08}
	\def\diff{0.36}
	\newcommand*{\ExtractCoordinate}[1]{\path (#1); \pgfgetlastxy{\XCoord}{\YCoord};}%
	\resizebox{\textwidth}{!}{
	\begin{tikzpicture}
	\matrix (first) [ampersand replacement=\&,row sep=0.7cm,column sep=0.2cm]
	{
		\node[boxes,minimum height=0.6cm] (dc) {distance to clique};
		\& \node[boxes,minimum height=0.6cm] (ce) {cluster editing};
		\& \node[boxes,minimum height=0.6cm] (vc) {vertex cover};
		\& \node  (FPTNoPolyKernel) {\Large\textbf{FPT}};
		\& \node[boxes,minimum height=0.6cm] (ml) {max leaf number};
		\&
		\\
			
		\node[boxes,minimum height=0.6cm,unk] (dcc)  {\disttoSec{co-cluster}} ;
		\& \node[boxes,minimum height=0.6cm] (dcl)  {cluster vertex deletion~(Thm.~\ref{thm:fptclusterdel})} ;
		\& \node[boxes,minimum height=0.6cm,unk] (ddp)  {\disttoSec{disjoint paths}} ;
		\& \node[boxes,minimum height=0.6cm] (td)  {tree-depth~\cite{DvorakK15}};
		\& \node[boxes,minimum height=0.6cm] (fes)  {\tworowsSec{feedback}{edge set} (Thm.~\ref{thm:fes-linear-kernel})};
		\& \node[boxes,minimum height=0.6cm] (bw)  {bandwidth} ;\\
			\node[boxes,minimum height=0.6cm,unk] (mcc)  {minimum clique cover};
		\& \node[boxes,minimum height=0.6cm,unk] (dcg)  {\disttoSec{cograph}} ;
		\& \node[boxes,minimum height=0.6cm,unk] (dig)  {\disttoSec{interval}} ;
		\& \node[boxes,minimum height=0.6cm,unk] (fvs)  {\tworowsSec{feedback vertex}{set}} ;
		\& \node[boxes,minimum height=0.6cm] (pw)  {pathwidth} ;
		\& \node[boxes,minimum height=0.6cm] (mxd) {\tworowsSec{maximum}{degree} (Prop.~\ref{prop:xp-delta})};\\
		   \node[boxes,minimum height=0.6cm,unk] (is)  {\tworowsSec{maximum}{independent set}};
		\&
		\& \node[boxes,minimum height=0.6cm,unk] (dch)  {\disttoSec{chordal}};
		\& \node[boxes,minimum height=0.6cm] (dbp)  {\disttoSec{bipartite} (Thm.~\ref{thm:bipartiteDegeneracy2})} ;
		\& \node[boxes,minimum height=0.6cm] (tw)  {treewidth~\cite{DvorakK15}} ;
		\& \node[boxes,minimum height=0.6cm,unk] (hindex)  {$H$-index} ;\\
		\node[boxes,minimum height=0.6cm,unk] (ds)  {domination number};
		\&
		\& \node[boxes,minimum height=0.6cm] (dperf) {\disttoSec{perfect}};
		\& 	\&
		\& \node[boxes,minimum height=0.6cm] (deg)  {degeneracy (Thm.~\ref{thm:bipartiteDegeneracy2})} ;\\
		\&  \node[boxes,minimum height=0.6cm] (mxdia)  {diameter (Thm.~\ref{thm:bequal2})};	\& 	\&
		\&\node[boxes,minimum height=0.6cm] (cn)  {\tworowsSec{chromatic}{number}} ;
		\& \node[boxes,minimum height=0.6cm] (avg) {average \mbox{degree}};\\
	};
		\draw[thick,-] (dc.south west) to [out=-115,in=115](mcc.north west);
		\draw[thick] (is) -- (mcc);
		\draw[thick] ($(dcl.north)-(1cm,0)$) -- (dc);
		\draw[thick] (dcc) -- (dc);
		\draw[thick] ($(dcc.north)+(0.5cm,0)$) -- ($(vc.south)-(1cm,0)$);
		\draw[thick] (dcl) -- (vc);
		\draw[thick] (ddp) -- (vc);
		\draw[thick] (dcl) -- (ce);
		\draw[thick] (ddp.north) -- (ml);
		\draw[thick] (bw) -- (ml);
		\draw[thick] (dcg) -- (dcc);
		\draw[thick] (dcg) -- (dcl);
		\draw[thick] (dig) -- (dcl);
		\draw[thick] (dig) -- (ddp);
		\draw[thick] (fvs) -- (ddp);
		\draw[thick] (fvs) -- (fes);
		\draw[thick] (pw) -- (ddp);
		\draw[thick] (pw) -- (td);
		\draw[thick] (td) -- (vc);
		\draw[thick] (ml)--(fes);
		\draw[thick] (pw) -- (bw);
		\draw[thick] (ds) -- (is);
		\draw[thick] (dch) -- (dig);
		\draw[thick] (dbp) -- (fvs);
		\draw[thick] (mxd) -- (bw);
		\draw[thick] (mxdia) -- (ds.south);
		\draw[thick] (mxdia) -- (dcg);
		\draw[thick] (tw) -- (fvs);
		\draw[thick] (tw) -- (pw);
		\draw[thick] (hindex) -- (mxd);
		\draw[thick] (deg) -- (tw);
		\draw[thick] (deg) -- (hindex);
		\draw[thick] (cn) -- (deg);
		\draw[thick] (cn) -- (dbp);
		\draw[thick] (deg)--(avg);
		\draw[thick]($(dperf.north)-(1cm,0)$)--(dcg);
		\draw[thick](dperf)--(dch.south);
		\draw[thick](dperf.north east)--(dbp.south);

		\node[below of=dperf,yshift=-6mm,xshift=18mm,text width=6cm] {\Large\textbf{NP-hard with constant parameter values}};
		\node[above of=bw,yshift=2mm,xshift=0mm] {\Large \textbf{XP}};

		\begin{pgfonlayer}{background}

			\draw[rounded corners,fill=red!40]
				($(ds.west |- mxdia.north) 		+ (-\eps, \diff)$)--
				($(dperf.west |- mxdia.north) 	+ (-\eps, \diff)$)--
				($(dperf.north west)		 	+ (-\eps, \diff)$)--
				($(dbp.west |- dperf.north) 	+ (-\eps, \diff)$)--
				($(dbp.north west) 				+ (-\eps, \diff)$)--
				($(dbp.north east) 				+ ( \eps, \diff)$)--
				($(dbp.south east) 				+ ( \eps, -\diff)$)--
				($(deg.north east) 				+ ( \eps, \diff)$)--
				($(avg.south east) 				+ ( \eps,-\diff)$)--
				($(ds.south west |- avg.south) 	+ (-\eps,-\diff)$)--
				cycle;

			\draw[rounded corners,top color=orange!40,bottom color=orange!40]
				($(bw.west |- ml.north)		+ (-\eps, \diff)$)--
				($(bw.east |- ml.north)		+ ( \eps, \diff)$)--
				($(mxd.south east)			+ ( \eps,-\diff)$)--
				($(mxd.south west)			+ (-\eps,-\diff)$)--
				($(tw.north east)			+ ( \eps,-\diff)$)--
				($(tw.south east)			+ ( \eps,-\diff)$)--
				($(tw.south west)			+ (-\eps,-\diff)$)--
				($(pw.north west)			+ (-\eps, \diff)$)--
				($(bw.west |- pw.north)		+ (-\eps, \diff)$)--
				cycle;

			\draw[line width=0.5pt,rounded corners,top color=green!40,bottom color=green!40]
				($(dc.north west |- ml.north)		+(-\eps, \diff)$)--
				($(ml.north east)		+( \eps, \diff)$)--
				($(fes.south east)	+( \eps,-\diff)$)--
				($(td.south west |- fes.south)	+(-\eps,-\diff)$)--
				($(td.north west |- vc.south)		+(-\eps,-\diff)$)--
				($(vc.south west)		+(-\eps,-\diff)$)--
				($(vc.west |- dcl.south)	+(-\eps,-\diff)$)--
				($(dcl.south west)	+(-\eps,-\diff)$)--
				($(dcl.west |- ce.south)		+(-\eps,-\diff)$)--
				($(dc.west |- ce.south)		+(-\eps,-\diff)$)--
				cycle;
		\end{pgfonlayer}
	\end{tikzpicture}
	}
	\vspace{-6mm}
	\caption{
		The parameterized complexity of SP-MVE with unit-length edges with respect to different graph parameters.
		Herein, ``distance to~$X$'' denotes the number of vertices that have to be deleted in order to transform the input graph into a graph of the graph class~$X$.
		For two parameters that are connected by a line, the upper parameter is weaker (that is, larger) than the parameter below \cite{KN12}.
		In the later sections we will only define the graph parameters that we directly 
		work with.
		Refer to \citet{SW13} for formal definitions of all parameters.
	}
	\label{fig:par-hier}
\end{figure}
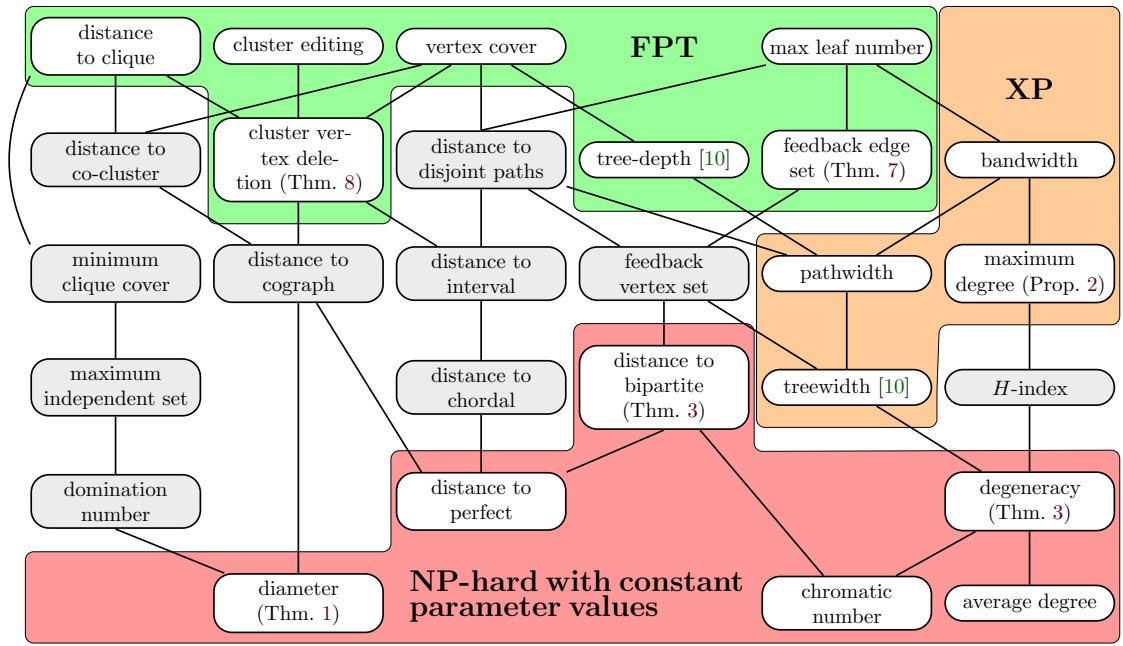
Moreover, towards the goal of spotting further fixed-parameter tractable special cases, it also suggests to look for reasonable parameter combinations.
In addition, \cref{tab:sp-mve-overview} overviews our exact and approximate complexity results for SP-MVE.
\begin{table}[t!]
	\setlength{\tabcolsep}{8pt}
	\centering	
	\label{tab:sp-mve-overview}
	\begin{tabular}{ccc}
		\toprule
								& $k$													& $\ell$ \\
		\midrule
		related to			&	XP								 					& NP-hard for~$b=2$ and~$\ell=9$ \\
		polynomial time	&									 					& $\ell$-approximation  \\
		\midrule
		\multirow{3}{*}{related to fpt time}			& $n/2^{O(\sqrt{\log n})}$-approximation	& $r(n)$-approximation for\\
				 		& for unit-length edges							& every increasing~$r$\\
								\cmidrule{2-3}
								& \multicolumn{2}{c}{fpt with respect to combined parameter~$(k,\ell)$} \\
		\bottomrule
	\end{tabular}
	\caption{Overview on the computational complexity classification of SP-MVE on $n$-vertex graphs. %
}
\end{table}
\Cref{fig:gcs-hier} summarizes our understanding of the complexity of \SPMVE with unit-length edges on several graph classes. %

\begin{figure}[t!]
	\centering
	\newcommand{\tworowsSec}[2]{#1 #2}
	\newcommand{\disttoSec}[1]{distance~to #1}
	\tikzstyle{boxes}=[draw,thick, rounded corners=3mm,text width=2.6cm,align=center,text opacity=1,fill opacity=1,fill=white]
	\tikzstyle{unk}=[fill=gray!15!white]
	\def\eps{0.08}
	\def\diff{0.18}
	\newcommand*{\ExtractCoordinate}[1]{\path (#1); \pgfgetlastxy{\XCoord}{\YCoord};}%
	\resizebox{0.8\textwidth}{!}
	{
	\begin{tikzpicture}
	\matrix (first) [ampersand replacement=\&,row sep=0.4cm,column sep=0.2cm]
	{
		\& \node[boxes,minimum height=0.6cm] (perfect) {perfect graphs};
		\&
		\& \node[boxes,minimum height=0.6cm] (planar) {planar graphs \cite{FHNN18}};
		
		\\
		
		\node[boxes,minimum height=0.6cm] (chor) {chordal graphs};
		\& 
		\& \node[boxes,minimum height=0.6cm] (bip) {bipartite graphs \\ (\cref{thm:bipartiteDegeneracy2})};
		\& \node[boxes,minimum height=0.6cm] (btw) {bounded tree\-width graphs \cite{DvorakK15}};
		\\
		\node[boxes,minimum height=0.6cm,unk] (int) {interval graphs};
		\& \node[boxes,minimum height=0.6cm] (split) {split graphs \\ (\cref{thm!split})};
		\& \node[boxes,minimum height=0.6cm] (cog) {cographs \\ (\Cref{prop:diamtwo})};
		\& \node[boxes,minimum height=0.6cm] (sp) {series-parallel graphs};
		\\
		\node[boxes,minimum height=0.6cm,fill=gray!5!white] (pint) {proper interval graphs~\cite{Stahlberg16}};
		\& \node[boxes,minimum height=0.6cm] (thr) {threshold graphs};
		\&
		\&
		 \node[boxes,minimum height=0.6cm] (outplanar) {outerplanar graphs};
		\\
		\& \node[boxes,minimum height=0.6cm] (complete) {complete graphs};
		\& \node[boxes,minimum height=0.6cm] (trees) {trees};
		\& \node[boxes,minimum height=0.6cm] (cactus) {cacti};
		\\
	};
		\draw[thick] (complete)-- (pint);
		\draw[thick,-] ($(trees.north)+(10mm,0)$) to [out=60,in=-60]($(bip.south)+(12mm,0)$);
		\draw[thick] (pint)-- (int);
		\draw[thick] (int)-- (chor);
		\draw[thick,-] ($(trees.north)+(11mm,0)$) to (outplanar);
		\draw[thick] (cactus)-- (outplanar);
		\draw[thick] (outplanar)-- (sp);
		\draw[thick,-] ($(sp.north east)-(1mm,0)$) to [out=60,in=-60]($(planar.south east)-(1mm,0)$);
		\draw[thick] (thr)-- (split);
		\draw[thick] (thr.north west)-- (int);
		\draw[thick] ($(thr.north east)-(3mm,0)$)-- ($(cog.south west)+(4mm,0)$);
		\draw[thick] (sp)-- (btw);
		\draw[thick] ($(cog.north west)+(2mm,0)$)-- (perfect);
		\draw[thick] (bip.north)-- (perfect);
		\draw[thick] (chor.north)-- (perfect);
		\draw[thick] (split)-- (chor);
		\draw[thick] (complete)-- (thr);
		
		\node[left of=perfect,yshift=1mm,xshift=-20mm] () {\large \textbf{NP-hard}};
		\node[below of=cog,yshift=-3mm,xshift=1mm,text width=3.25cm] () {\large \textbf{polynomial-time solvable}};

		\begin{pgfonlayer}{background}

			\draw[rounded corners,fill=red!40]
				($(chor.west |- planar.north) 		+ ( -\eps, \diff)$)--
				($(planar.east |- planar.north) 	+ ( \eps, \diff)$)--
				($(planar.south east) 				+ ( \eps,-\diff)$)--
				($(bip.east |- planar.south) 		+ ( \eps,-\diff)$)--
				($(bip.south east) 					+ ( \eps,-\diff)$)--
				($(split.east |- bip.south) 		+ ( \eps,-\diff)$)--
				($(split.south east) 				+ ( \eps,-\diff)$)--
				($(split.south west) 				+ (-\eps,-\diff)$)--
				($(split.west |- chor.south)		+ (-\eps,-\diff)$)--
				($(chor.south west) 				+ (-\eps,-\diff)$)--
				cycle;
				
			\draw[line width=0.5pt,rounded corners,top color=green!40,bottom color=green!40]
				($(thr.west |- thr.north)	+(-\eps, \diff)$)--
				($(cog.west |- thr.north)	+(-\eps, \diff)$)--
				($(cog.west |- sp.north)	+(-\eps, \diff)$)--
				($(sp.north west)			+(-\eps, \diff)$)--
				($(btw.west |- sp.north)	+(-\eps, \diff)$)--
				($(btw.north west)			+(-\eps, \diff)$)--
				($(btw.north east)			+( \eps, \diff)$)--
				($(btw.south east)			+( \eps,-\diff)$)--
				($(cactus.south east)		+( \eps,-\diff)$)--
				($(complete.south west)		+(-\eps,-\diff)$)--
				($(thr.west |- complete.south)			+(-\eps,-\diff)$)--
				cycle;
				
			\draw[line width=0.5pt,dashed,rounded corners,top color=green!20,bottom color=green!20]
				($(pint.south west) 		+ (-0.25*\eps,-0.5*\diff)$)--
				($(pint.north west) 		+ (-0.25*\eps,0.5*\diff)$)--
				($(pint.north east) 		+ (0.25*\eps,0.5*\diff)$)--
				($(pint.south east) 		+ (0.25*\eps,-0.5*\diff)$)--
				cycle;
		\end{pgfonlayer}
	\end{tikzpicture}
	}
	\caption{
		Computational complexity of SP-MVE with unit-length edges for some graph classes.
		For SP-MVE with unit-length edges on proper interval graphs, we conjecture that it is solvable in polynomial time.
	}
	\label{fig:gcs-hier}
\end{figure}
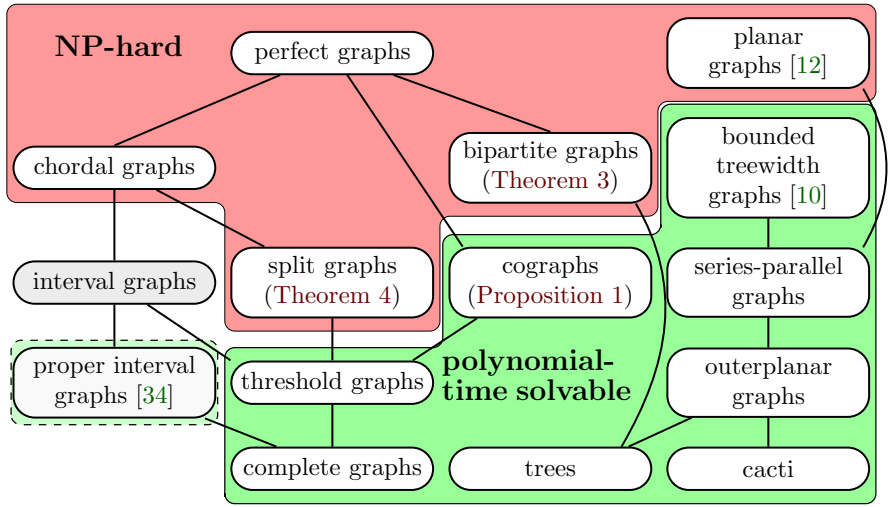

\paragraph*{Organization of the paper.}
After introducing
some preliminaries in \cref{sec_prelim}, we prove in \cref{sec_complexity} our NP-hardness results. In \cref{sec_polynomial},
we present our polynomial-time solvable special cases.
In \cref{sec_fpt}, we provide  parameterized and approximation algorithms for SP-MVE.
Conclusions and open questions are provided in \cref{sec_conclusion}.

\section{Preliminaries}\label{sec_prelim}

For an undirected graph~$G=(V,E)$ we set~$n:=|V|$ and~$m:=|E|$. %
A path~$P$ of length~$r-1$ in~$G$ is a sequence of distinct vertices~$P = v_1$-$v_2$-$\ldots$-$v_r$ with~$\{v_i,v_{i+1}\} \in E$ for all~$i \in \{1,\ldots, r-1\}$; the vertices~$v_1$ and~$v_r$ are the endpoints of the path.
For~$1 \le i < j \le r$, we set~$v_iPv_j$ to be the subpath of~$P$ starting in~$v_i$ and ending in~$v_j$, formally~$v_iPv_j := v_i$-$v_{i+1}$-$\ldots$-$v_{j}$.
For~$i = 1$ or~$j=r$ we omit the corresponding endpoint, that is, we set~$Pv_j := v_1Pv_j$ and~$v_iP := v_iPv_r$.%
For~$u,v \in V$, a~$uv$-path~$P$ is a path with endpoints~$u$ and~$v$.
The distance between~$u$ and~$v$ in~$G$, denoted by~$\dist_G(u,v)$, is the length of a shortest~$uv$-path.
The diameter of~$G$ is the length of the longest shortest path in~$G$.

For $v\in V$ let $N_G(v)$ be the set of neighbors of~$v$ and 
let $N_G[v]=N_G(v)\cup \{v\}$ be~$v$'s closed neighborhood.
Two vertices~$u,v \in V$ are called \emph{true twins} if~$N_G[u] = N_G[v]$ and \emph{false twins} if~$N_G(u) = N_G(v)$ but~$N_G[u] \neq N_G[v]$; they are called \emph{twins} if they are either true or false twins.
We denote by $G-S$ the graph obtained from $G$ by removing the edge subset $S \subseteq E$.
For~$s,t \in V$, an edge subset~$S$ is called $st$-cut if~$G-S$ contains no~$st$-path.
For~$V' \subseteq V$ let~$G[V']$ denote the subgraph induced by~$V'$.
For~$E' \subseteq E$ let~$G[E']$ denote the subgraph consisting of all endpoints of edges in~$E'$ and the edges in~$E'$.

\paragraph{Parameterized complexity.}
A parameterized problem consisting of input instance~$I$ and 
parameter~$k$ is called \emph{fixed-pa\-ra\-me\-ter tractable} (fpt) if there is an algorithm that decides any instance $(I,k)$ %
in $f(k)\cdot|I|^{O(1)}$ time for some computable function~$f$ solely depending on~$k$, where~$|I|$ denotes the size of~$I$. 
On the contrary, the parameterized complexity class~XP contains 
all parameterized problems that can be solved in $|I|^{f(k)}$~time; 
in other words,
membership in~XP means polynomial-time solvability when the parameter value is 
a constant.

A core tool in the development of fixed-parameter tractability results is polynomial-time preprocessing by data reduction, called \emph{kernelization}~\cite{Kra14,GN07SIGACT}.
Here, the goal is to transform a given problem instance~$(I, k)$ in polynomial time into an equivalent instance~$(I', k')$ whose size is upper-bounded by a function of~$k$.
That is, $(I, k)$ is a yes-instance if and only if~$(I', k')$ with $|I'|,k'\leq g(k)$ for some function~$g$ is a yes-instance.
Thus, such a transformation is a polynomial-time self-reduction with the constraint that the reduced instance is ``small'' (measured by~$g(k)$).
If such a transformation exists, then $I'$ is called \emph{(problem) kernel} of size~$g(k)$.

\paragraph{Approximation.}
Given an NP optimization problem and an instance $I$ of this problem, we use 
$\opt(I)$ to denote the optimum value of $I$ and $\val(I,S)$ to denote the value of a feasible solution $S$ of instance $I$.
The {\em approximation ratio\/} of $S$ (or {\it approximation factor}) is $r(I,S)=\max\left\{\frac{\val(I,S)}{\opt(I)}, \frac{\opt(I)}{\val(I,S)}\right\}.$
For a function $\rho$, an algorithm~$\mathcal{A}$ is a {\it $\rho(|I|)$-approximation\/} if for every instance~$I$ of the problem, it returns a solution $S$ such that $r(I,S) \leq \rho(|I|)$.
If the problem comes with a parameter~$k$ and the algorithm~$\mathcal{A}$ runs in~$f(k) \cdot |I|^{O(1)}$ time, then $\mathcal{A}$ is called \emph{parameterized $\rho(|I|)$-approximation}.

\section{NP-hardness results} \label{sec_complexity}

In this section, we provide several hardness results for restricted variants of \SPMVE. 
We start by adapting a reduction idea due to \citet{KBBEGRZ08} for the vertex deletion variant of SP-MVE.
We prove that SP-MVE is NP-hard even for constant values of $b$, $\ell$, and the diameter of the input graph.

\begin{theorem}\label{thm:bequal2}
	\SPMVE is NP-hard, even for unit-length edges, $b=2$, $\ell = 9$, and diameter~$8$.
\end{theorem}
{\begin{proof}
	As \citet[Theorems 8 and 11]{KBBEGRZ08}, we reduce from the NP-hard \cite[GT1]{GJ79} problem~\VC on tripartite graphs, where the question is, given a tripartite graph~$G=(V=V_1 \uplus V_2 \uplus V_3,E)$ and an integer~$h\geq0$, whether there is a subset~$V'\subseteq V$ with~$|V'|\leq h$ such that $G[V\setminus V']$ contains no edge.
	While the fundamental approach remains the same, the technical details when moving their vertex deletion scenario to our edge deletion scenario change to quite some extent.
	We refrain from a step-by-step comparison.
	Given a \VC instance~$(G,h)$ with~$G = (V_1 \uplus V_2 \uplus V_3,E)$ being a tripartite graph on $n$ vertices, we construct an SP-MVE instance $I'=(G',k,\ell)$ as follows.
	First, let~$k := h$ and~$\ell := 9$.
	The graph $G'=(V',E')$ contains vertices $V' = V_1 \uplus V_2\uplus V_3\uplus V_2' \uplus \{s,t\}$, where $s$ and $t$ are two new vertices, and for each $v\in V_2$ we add a copy $v'\in V_2'$.
	
	Before describing the edge set~$E'$, we introduce edge-gadgets.
	Here, by adding a length-$\alpha$ \emph{edge-gadget}~$e_{u,v}$, $\alpha \ge 2$, from the vertex~$u$ to vertex~$v$, we mean to add $n$~vertex-disjoint paths of length~$\alpha-2$ and to make~$u$ adjacent to the first vertex of each path and~$v$ adjacent to the last vertex of each path.
	If $\alpha = 2$, then each path is just a single vertex which is at the same time the first and last vertex.
	The idea behind this is that one will never delete edges in an  edge-gadget.

	We add the following edges and edge-gadgets to~$G'$ (see \cref{fig:np-hardness} for a schematic representation of the constructed graph).
	\begin{figure}[t]
		\centering
		\tikzstyle{knoten}=[circle,draw,fill=black!20,minimum size=5pt,inner sep=2pt]
		\tikzstyle{knoten-set}=[ellipse,draw,fill=black!05,minimum height=2cm,minimum width=0.5cm]
		\begin{tikzpicture}[>=stealth',draw=black!75]
			\node[knoten,label=below:{$s$}] (s) at (0,0) {};
			\node[knoten,label=below:{$t$}] (t) at (8,0) {};

			\node[knoten-set,label=below:{$V_1$}] (V1) at (2, 0) {};
			\path (s) edge[-] node[auto,swap] {1} (V1);

			\node[knoten-set,label=above:{$V_2$}] (V2)  at (3.5, 1.75) {};
			\node[knoten-set,label=above:{$V'_2$}] (VV2) at (4.5, 1.75) {};
			\path (VV2) edge[-] node[auto,swap] {1} (V2);
			\path (V2) edge[ultra thick,-,bend right] node[auto,swap] {4} (s);
			\path (t) edge[ultra thick,-,bend right] node[auto,swap] {4} (VV2);

			\node[knoten-set,label=below:{$V_3$}] (V3) at (6, 0) {};
			\path (t) edge[-] node[auto,swap] {1} (V3);

			\path (V1) edge[ultra thick,-] node[auto,swap] {2} (V2);
			\path (V1) edge[ultra thick,-] node[auto,swap] {5} (V3);
			\path (VV2) edge[ultra thick,-] node[auto,swap] {2} (V3);

			\foreach \i in {1,...,4}{
				\node[knoten] () at (2, 0.3 * \i - 0.75) {};
				\node[knoten] () at (6, 0.3 * \i - 0.75) {};
				\node[knoten] () at (3.5, 0.3 * \i + 1) {};
				\node[knoten] () at (4.5, 0.3 * \i + 1) {};
			}
		\end{tikzpicture}
		\caption{
			A schematic representation of the graph~$G'$ constructed from the tripartite graph~$G=(V_1 \uplus V_2 \uplus V_3,E)$.
			The vertices are grouped into the described sets.
			The edges in the picture correspond to edge sets in~$G'$ and cover the incidence structure of the displayed vertices in~$G'$.
			A bold edge indicates an edge-gadget and the corresponding number denotes its length.
		}
		\label{fig:np-hardness}
	\end{figure}
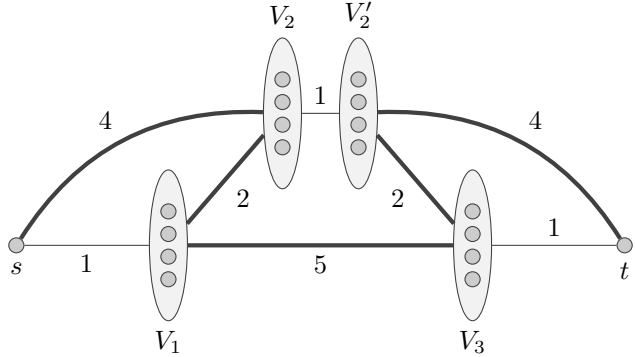
	For each vertex~$v \in V_2$ we add the edge~$\{v,v'\}$ %
        between $v$ and its copy~$v'$.
	For each vertex $v\in V_1$, we add the edge~$\{s,v\}$, and for each vertex~$v\in V_3$, we add the edge~$\{v,t\}$.
	We also add the following edge-gadgets:
	For each edge $\{u,v\} \in (V_1 \times V_2) \cap E$ we add the edge-gadget $e_{u,v}$ of length two, for each edge $\{u,v\} \in (V_2\times V_3) \cap E$ we add the edge-gadget $e_{u',v}$ of length two, where~$u'\in V_2'$ is the copy of~$u$, and for each edge $\{u,v\} \in (V_1\times V_3) \cap E$ we add the edge-gadget $e_{u,v}$ of length five.
	Furthermore, we add edge-gadgets of length four between $s$ and every vertex $v\in V_2$ and between $t$ and every vertex $v'\in V_2'$.
	Observe that we have $\dist_{G'}(s,t)=7$ and thus~$b = \ell - \dist_{G'}(s,t) = 2$. %

	We now show that $G$ has a vertex cover of size at most $h$ if and only if deleting~$k=h$ edges in~$G'$ results in~$s$ and~$t$ having distance at least~$\ell = 9$.

	``$\Rightarrow:$''
	Let $V'' \subseteq V$ be a vertex cover of size at most~$h$ in $G$.
	Consider the edge sets~$E_1'':=\{\{s,v\} : v\in V_1 \cap V''\}$, $E_2'':=\{\{\{v,v'\} : v\in V_2 \cap V'',v'\in V_2'\text{~copy of~}v\}$, and $E_3'':=\{\{v,t\} : v\in V_3 \cap V''\}$.
	We claim that for the set $$E''= E_1'' \cup E_2''\cup E_3''$$ it holds that~$\dist_{G'-E''}(s,t) \geq 9$  and~$|E''| = |V''| \le h$.
	Clearly, $|E''| = |V''| \le h$.
	Suppose towards a contradiction that~$\dist_{G'-E''}(s,t) < 9$.
	Let~$P$ be an $st$-path of length less than~nine.
	Observe that~$P$ contains an edge connecting~$s$ with some vertex in~$V_1$ or an edge connecting~$t$ with some vertex in~$V_3$.
	We discuss only the first case, as the second follows by symmetry.
	
	Let~$P$ contain an edge connecting~$s$ with vertex~$u$ in~$V_1$.
	Path~$P$ contains either~(i) a subpath of length~three to vertex in~$V_2'$, or~(ii) a subpath of length~five to a vertex in~$V_3$.
	
	\emph{Case (i)}:
	Let~$u-a_1-a_2-v-v'$, with~$v\in V_2$ and~$v'\in V_2$ the copy of~$v$, be a subpath of~$P$, where~$a_1,a_2$ are vertices in an edge-gadget~$e_{u,v}$.
	Then~$\{u,v\}\in E$ and~$u,v\not\in V''$, as~$\{s,u\}\not\in E_1''$ and~$\{v,v'\}\not\in E_2''$, contradicting that~$V''$ is a vertex cover of~$G$.
	
	\emph{Case (ii)}:
	Let~$u-a_1-\ldots-a_5-v$ with~$v\in V_3$ be a subpath of~$P$, where~$a_1,\ldots,a_5$ are vertices in an edge-gadget~$e_{u,v}$.
	As~$P$ is of length less than~$9$, it follows that~$P=s-u-a_1-\ldots-a_5-v-t$.
	Then~$\{u,v\}\in E$ and~$u,v\not\in V''$, contradicting that~$V''$ is a vertex cover of~$G$.

	``$\Leftarrow:$''
	Let $E'' \subseteq E'$ be a set of edges such that~$\dist_{G'-E''}(s,t) \ge 9$ and~$|E''| \le h$.
	If $E''$ contains edges from an edge-gadget $e_{u,v}$, then it must contain at least $n$ edges from this gadget in order to have a chance to increase the solution value.
	Therefore, since $h<n$, we can assume that~$E''$ does not contain any edge contained in an edge-gadget.
	Thus, $E'' \subseteq (\{s\}\times V_1) \cup (V_2 \times V'_2) \cup (V_3 \times \{t\})$.
	We construct a vertex cover~$V''$ for~$G$ as follows:
	For each edge~$\{s,v\} \in E''$ it follows that~$v \in V_1$ and we add~$v$ to~$V''$.
	Similarly, for each edge~$\{v,t\} \in E''$ it follows that~$v \in V_3$ and we add~$v$ to~$V''$.
	Finally, for each edge~$\{v,v'\} \in E'' \cap (V_2 \times V_2')$, we add~$v$ to~$V''$.

	Suppose towards a contradiction, that~$V''$ is not a vertex cover in~$G$, that is, there exists an edge~$\{u,v\} \in E$ with~$u,v \notin V''$.
	If~$v \in V_1$ and~$u \in V_2$, then the $st$-path $s$-$v$-$u$-$u'$-$t$ of length~$8 < \ell$ is contained in~$G'-E''$.
	If~$v \in V_1$ and~$u \in V_3$, then the $st$-path $s$-$v$-$u$-$t$ of length~$7 < \ell$ is contained in~$G'-E''$.
	Finally, if~$v \in V_2$ and~$u \in V_3$, then the $st$-path $s$-$v$-$v'$-$u$-$t$ of length~$8 < \ell$ is contained in~$G'-E''$.
	Each of the three cases contradicts the assumption that~$\dist_{G'-E''}(s,t) \ge 9$.
\end{proof}}
\citet{BaierEHKKPSS10} showed that \SPMVE is polynomial-time solvable for the special case of~$b=1$.
\Cref{thm:bequal2} shows that this result cannot be extended to larger values of~$b$.
Regarding the diameter of the input graph, the statement of \Cref{thm:bequal2} will be strengthened later:
Considering the problem with unit-length edges, we show that it remains NP-hard on graphs of diameter three (\Cref{thm!split}), while it becomes polynomial-time solvable on graphs of diameter two (\Cref{prop:diamtwo}).
For arbitrary edge lengths, we show that the problem remains NP-hard on graphs of diameter one~(\Cref{thm:cliques-np-hard}).

When allowing length zero edges, \citet{KBBEGRZ08} stated that it is NP-hard to approximate \kmve within a factor smaller than two.
We consider in this paper only positive edge lengths and, by adapting the construction given in the above proof by considering edge-gadgets of lengths polynomial in~$n$ (with high degree), we obtain the following.

\begin{theorem}\label{thm:apx-hardness}
Unless P${}={}$NP, 
	\kmve is  not $4/3-1/\poly(n)$-approxi\-ma\-ble in polynomial time, even for unit-length edges. 
\end{theorem}
{\begin{proof}
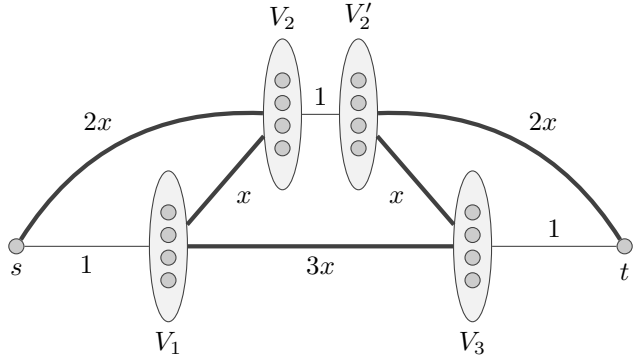
\begin{figure}[t]
	\centering
	\tikzstyle{knoten}=[circle,draw,fill=black!20,minimum size=5pt,inner sep=2pt]
	\tikzstyle{knoten-set}=[ellipse,draw,fill=black!05,minimum height=2cm,minimum width=0.5cm]
	\begin{tikzpicture}[>=stealth',draw=black!75]
		\node[knoten,label=below:{$s$}] (s) at (0,0) {};
		\node[knoten,label=below:{$t$}] (t) at (8,0) {};

		\node[knoten-set,label=below:{$V_1$}] (V1) at (2, 0) {};
		\path (s) edge[-] node[auto,swap] {1} (V1);

		\node[knoten-set,label=above:{$V_2$}] (V2)  at (3.5, 1.75) {};
		\node[knoten-set,label=above:{$V'_2$}] (VV2) at (4.5, 1.75) {};
		\path (VV2) edge[-] node[auto,swap] {1} (V2);
		\path (V2) edge[ultra thick,-,bend right] node[auto,swap] {$2x$} (s);
		\path (t) edge[ultra thick,-,bend right] node[auto,swap] {$2x$} (VV2);

		\node[knoten-set,label=below:{$V_3$}] (V3) at (6, 0) {};
		\path (t) edge[-] node[auto,swap] {1} (V3);

		\path (V1) edge[ultra thick,-] node[auto,swap] {$x$} (V2);
		\path (V1) edge[ultra thick,-] node[auto,swap] {$3x$} (V3);
		\path (VV2) edge[ultra thick,-] node[auto,swap] {$x$} (V3);

		\foreach \i in {1,...,4}{
			\node[knoten] () at (2, 0.3 * \i - 0.75) {};
			\node[knoten] () at (6, 0.3 * \i - 0.75) {};
			\node[knoten] () at (3.5, 0.3 * \i + 1) {};
			\node[knoten] () at (4.5, 0.3 * \i + 1) {};
		}
	\end{tikzpicture}
	\caption{
			A schematic representation of the graph~$G'$ constructed from the tripartite graph~$G=(V_1 \uplus V_2 \uplus V_3,E)$.
			The vertices are grouped to the used sets.
			The edges in the picture correspond to edge sets in~$G'$ and cover the incidence structure of the displayed vertices in~$G'$.
			A bold edge indicates an edge-gadget and the corresponding number denotes its length.
	}
	\label{fig:approx-hardness}
\end{figure}

	We construct a gap-reduction~\cite{AL96} from  \VC on tripartite graphs to \kmve.
	More specifically, we use a gap-reduction from a decision problem to a maximization problem. A decision problem $\Pi$ is called \emph{gap-reducible} to a maximization problem~$\Pi'$  with gap~$\rho(|I|) >1$ if
	for any instance~$I$ of~$\Pi$ we can construct an instance~$I'$ of~$Q$ in polynomial time  while satisfying the following properties for some function~$c:\N\to\R\cap (0,+\infty)$.
	\begin{itemize}
		\item If~$I$ is a yes-instance, then $ \opt(I^\prime) \geq c(|I|)$.
		\item If $I$ is a no-instance,  then $\opt(I^\prime) < \frac{c(|I|)}{\rho(|I|)}$.
	\end{itemize}
	The idea behind a $gap$-reduction is that if $\Pi$ is NP-hard then $\Pi^\prime$ is not approximable within a factor $\rho$ provided that P $\neq $ NP.

	Starting with an instance $(G=(V,E),h)$ of \VC on tripartite graphs we construct an instance $I'=(G'=(V',E'), k, s, t)$ of \kmve as in the proof of \cref{thm:bequal2}.  We only change some lengths as follows (see also \cref{fig:approx-hardness}):
	For each edge $\{u,v\} \in (V_1 \times V_2) \cap E$ we add the edge-gadget~$e_{u,v}$ of length~$x$, for each edge~$\{u,v\} \in (V_2\times V_3) \cap E$ we add the edge-gadget~$e_{u',v}$ of length~$x$, where~$u'\in V_2'$ is the copy of~$u$, and for each edge $\{u,v\} \in (V_1\times V_3) \cap E$ we add the edge-gadget $e_{u,v}$ of length $3x$.
	We add edge-gadgets of length $2x$ between $s$ and every vertex $v\in V_2$ and between $t$ and every vertex $v'\in V_2'$. 
	The value $x$ could be any polynomial function in~$|V|=n$.
	Observe that we have $\dist_{G'}(s,t)\leq 3x+2$.

	We now show that if $G$ has a vertex cover of size at most~$h$, then $\opt(I')\geq 4x+1$, otherwise $\opt(I') \leq 3x+2$.

	Let $V'' \subseteq V$ be a vertex cover of size at most~$h$ in~$G$.
	It is not hard to verify (see~proof of \cref{thm:bequal2}) that for the set~$E''= \{\{s,v\} : v\in V_1 \cap V''\} \cup \{\{v,v'\} : v\in V_2 \cap V'',v'\in V_2'\text{~copy of~}v\}\cup \{\{v,t\} : v\in V_3 \cap V''\}$ it holds that~$\dist_{G'-E''}(s,t) = 4x+1$  and~$|E''| = |V''| \le h$.

	Suppose now that $G$ has no vertex cover of size  $h$.  Let $E'' \subseteq E'$ be a set of $h$~edges.
	As in the proof of \cref{thm:bequal2}, we can assume that~$E''$ does not contain any edge from an edge-gadget.
 	Thus~$E'' \subseteq (\{s\}\times V_1) \cup (V_2 \times V'_2) \cup (V_3 \times \{t\})$.
	We construct a vertex set~$V''$ for~$G$ as follows:
	For each edge~$\{s,v\} \in E''$, we add~$v$ to~$V''$ and
	  for each edge~$\{v,t\} \in E''$, we add~$v$ to~$V''$.
	Finally, for each edge~$\{v,v'\} \in E'' \cap (V_2 \times V_2')$, we add~$v$ to $V''$.

	Since~$V''$ is not a vertex cover in~$G$,   there exists an edge~$\{u,v\} \in E$ with~$u,v \notin V''$.
	If~$v \in V_1$ and~$u \in V_2$, then the $st$-path $s$-$v$-$u$-$u'$-$t$ of length~$3x+2$ is contained in~$G'-E''$.
 	If~$v \in V_1$ and~$u \in V_3$, then the $st$-path $s$-$v$-$u$-$t$ of length~$3x+2$ is contained in~$G'-E''$.
	Finally, if~$v \in V_2$ and~$u \in V_3$, then the $st$-path $s$-$v$-$v'$-$u$-$t$ of length~$3x+2$ is contained in~$G'-E''$.

	Since \VC is NP-hard on tripartite graphs \cite[GT1]{GJ79}, \kmve is not $\frac{4x+1}{3x+2}=4/3-1/\poly(n)$-approximable in polynomial time.
\end{proof}}
Concerning special graph classes, we can show that the problem remains NP-hard on restricted bipartite graphs.
To formulate our result, we need the graph parameter degeneracy. A graph~$G$ has \emph{degeneracy}~$d$ if every subgraph of~$G$ contains a vertex of degree at most~$d$.
By subdividing every edge, we obtain the following.
\begin{theorem}\label{thm:bipartiteDegeneracy2}
	\SPMVE is NP-hard, even for bipartite graphs with degeneracy two, unit-length edges, $b = 4$, $\ell = 18$, and diameter~$8$.
\end{theorem}
{\begin{proof}
	We provide a self-reduction from SP-MVE with unit-length edges with~$b=2$, $\ell = 9$, and diameter~$8$.
	Let~$I = (G=(V,E),k,\ell,s,t)$ be the given SP-MVE instance.
	We construct an instance~$I' = (G',k,2 \ell,s,t)$ where~$G'$ is obtained from~$G$ by subdividing all edges, that is, each edge is replaced by a path of length two.
	The correctness of the reduction is easy to see as any minimal solution contains at most one edge of each of the introduced induced paths of length two.
	Clearly, $I'$ can be computed in polynomial time.
	Furthermore, $G'$ is bipartite and has degeneracy two.
\end{proof}}

We next prove that \SPMVE remains NP-hard on split graphs.
A split graph is a graph whose vertex set can be partitioned into a clique and an independent set. 
Observe that a split graph has diameter at most three.
Thus, the next theorem also shows NP-hardness on diameter-three graphs.

\begin{theorem}\label{thm!split}
	\SPMVE is NP-hard on split graphs, even for unit-length edges.
\end{theorem}
\begin{proof}
 We reduce from SP-MVE on general graphs.
 Let $\I':=(G=(V,E),s,t,k,\ell)$ be an instance of SP-MVE, recall $n=|V|$.
 We obtain the graph~$G'=(V',E')$ from~$G$ by subdividing each edge of~$G$, and subsequently turning $V(G)$ into a clique.
 Formally, the graph $G'=(V',E')$ is defined through %
 \begin{align*}
  V' &:= V\cup (W:=\{w^{\{u,v\}}_j\mid \{u,v\}\in E, j\in[n^2]\}), \\
  E' &=  \binom{V}{2} \cup \left\{\{u,w^{\{u,v\}}_j\},\{v,w^{\{u,v\}}_j\}\mid \{u,v\}\in E, j\in[n^2]\right\} .
 \end{align*}
 Observe that $G'$ is a split graph since $G'[W]$ forms an independent set and $G'[V]$ forms a clique.
 Let $\I:=(G',s,t,k',\ell')$ be an instance of SP-MVE on split graphs with $k'=\binom{n}{2}+k\cdot n^2$ and $\ell':=2\ell$.
 We show that $\I$ is a yes-instance if and only if $\I'$ is a yes-instance.

 Let $\I$ be a yes-instance.
 Let $S\subseteq E(G)$ be such that $G-S$ has no $st$-path of length smaller than $\ell$.
 We claim that $G'-S'$ with $$S':=\binom{V}{2} \cup \left\{\{u,w^{\{u,v\}}_j\}\mid \{u,v\}\in S, j\in[n^2]\right\}$$ does not have an $st$-path of length smaller than $\ell'$.

 Note that $|S'|\leq \binom{n}{2} + k\cdot n^2$.
 Suppose that there is an $st$-path $P'$ in $G'-S'$ with $|P'|<2\ell$.
 Then the vertices in $P'$ alternate between the vertices in $V$ and $W$.
 By construction, if $\{v,w\},\{w,u\}\in E(P')$ with $u,v\in V$ and $w\in W$, then the edge~$\{u,v\}$ is present in~$G-S$.
 Hence, consider the $st$-path~$P$ in $G-S$ obtained from $P'$ by restricting $P'$ to $V$.
 It follows that $|P|=|P'|/2<\ell$, a contradiction to the choice of~$S$.
 Thus $\I'$ is a yes-instance.

 Conversely, let $\I'$ be a yes-instance.
 Let $S'\subseteq E(G')$ be minimal such that $G'-S'$ has no $st$-paths of length smaller than $\ell'$.
 We claim that $G-S$ with $$S:=\left\{\{u,v\}\mid \exists w^{\{u,v\}}_j\in W, e\in S' : w^{\{u,v\}}_j\in e\right\}$$ does not have an $st$-path of length smaller than $\ell$.
 If $\{u,w^{\{u,v\}}_j\}\in S'$ for some $u,v\in V$ and $j\in [n^2]$, then for all $i\in[n^2]$, $w^{\{u,v\}}_i$ is incident to exactly one edge in $S'$ since $S'$ is minimal
 (otherwise $S'\backslash \{u,w^{\{u,v\}}_j\}$ is a smaller solution).
 Together with~$|S'|\leq \binom{n}{2}+k\cdot n^2<(k+1)\cdot n^2$ it follows that $|S|<k+1$.
 Suppose there is an $st$-path $P$ in $G-S$ with $|P|<\ell$.
 Then for each edge $\{u,v\}\in E(P)$, there is a $j\in[n^2]$ such that $\{u,w^{\{u,v\}}_j\},\{v,w^{\{u,v\}}_j\}\not\in S'$.
 We construct an $st$-path $P'$ in $G'-S'$ from $P$ by replacing each edge $\{u,v\}\in E(P)$ by two edges $\{u,w^{\{u,v\}}_j\},\{v,w^{\{u,v\}}_j\}\not\in S'$ for some $j\in[n^2]$.
 Then $|P'|\leq 2\cdot |P|<2\cdot \ell$, a contradiction to the choice of~$S'$.
 Thus $\I$ is a yes-instance.
\end{proof}

Note that SP-MVE can be solved on complete graphs with unit-length edges in polynomial time.
If $\ell=1$, then the instance is trivially a yes-instance.
If $\ell=2$, one edge deletion is necessary to obtain the desired distance.
If $\ell>2$, then observe that for each vertex~$v \in V \setminus \{s,t\}$ the path~$s-v-t$ has length two and all these paths are edge-disjoint.
Hence, to increase the distance between~$s$ and~$t$ to three, we have to delete~$n-1$ edges (the edge~$\{s,t\}$ and one edge in each of the~$n-2$ paths of length two).
However, with~$n-1$ edge deletions, one can delete all edges incident to~$s$ and disconnect~$s$ from~$t$, so this solution works for all~$\ell > 2$.
Thus, if~$\ell > 2$, then the instance is a yes-instance if and only if the number of edge-deletion is at least~$n-1$.

As soon as one deals with arbitrary edge lengths, however, the problem becomes NP-hard even on complete graphs.

\begin{theorem}\label{thm:cliques-np-hard}
	\SPMVE remains NP-hard on complete graphs.
\end{theorem}

\begin{proof}
 We reduce from SP-MVE on general graphs.
 Let $\I:=(G=(V,E),s,t,k,\ell)$ be an instance of SP-MVE (w.l.o.g.~let $G$ not contain isolated vertices).
 Let $G'$ be the graph obtained from~$G$ by adding the edge set $E':=\{\{v,w\}|\{v,w\}\not\in E\}$ and assigning length~$\tau(e):=\ell+1$ to each edge $e\in E'$.
 Observe that $G'$ is a complete graph.
 We claim that $\I':=(G',s,t,k,\ell)$ is a yes-instance of SP-MVE if and only if $\I$ is a yes-instance of SP-MVE.

 By construction, $G$ is isomorphic to $G'[E(G)]$.
 This implies that for any $S\subseteq E(G)$, there is a bijection between the set of $st$-paths in~$G-S$ and the set of $st$-paths in~$G'[E(G)]-S$.
 Observe that every $st$-path in~$G'$ using an edge in $E(G')\backslash E(G)$ has length greater than~$\ell$.
 Hence, if there is an $S\subseteq E(G)$ such that there is no $st$-path in $G-S$ of length smaller than $\ell$, then there is no $st$-path in $G'-S$ of length smaller than $\ell$, and vice versa.
\end{proof}

\section{Polynomial-time algorithms} \label{sec_polynomial}

In this section, we present three polynomial-time algorithms for special cases of SP-MVE.

We start with considering instances of SP-MVE on series-parallel graphs with $s$ and $t$ being the natural two terminals of the underlying two-terminal graph.
Here, a two-terminal graph is a triplet containing a graph and two distinct vertices of the graph (the terminals).
Every two-terminal series-parallel graph can be constructed by a sequence of parallel and serial compositions starting from single edges where the endpoints of an edge are the two terminals.
Given two two-terminal series-parallel graphs $G_1$ and $G_2$ with terminals $s_1,t_1$ and $s_2,t_2$ respectively, then
\begin{compactenum}
 \item $G$ is a \emph{serial composition} of $G_1$ and $G_2$ with terminals $s_1,t_2$ if $G$ is the disjoint union of $G_1$ and $G_2$ where $t_1$ is identified with $s_2$.
 \item $G$ is a \emph{parallel composition} of $G_1$ and $G_2$ with terminals $s,t$ if $G$ is the disjoint union of $G_1$ and $G_2$ where $s_1$ is identified with $s_2$ and $t_1$ is identified with $t_2$.
\end{compactenum}

Moreover, we can construct for each two-terminal series-parallel graph~$G$ a so-called \emph{sp-tree} in linear time~\cite{ValdesTL82,BodlaenderF01}, a binary rooted tree representing the serial and parallel composition of two-terminal series-parallel graphs to obtain~$G$.
Herein, every leaf~$\alpha$ of the sp-tree is identified with an edge, and the label~$\lambda(\alpha)$ of the leaf~$\alpha$ is the set of the endpoints of the edge.
Moreover, each inner node~$\alpha$ of the sp-tree is labeled by either $\lambda(\alpha)=\mathbf{S}$ or $\lambda(\alpha)=\mathbf{P}$, representing a serial or parallel composition, respectively.

\begin{theorem}\label{thm!spgraphs}
	\textsc{Min-Cost}-\SPMVE can be solved in $O(m\cdot\ell^2)$~time on two-terminal series-parallel graphs with $s$ and $t$ being the two terminals.
\end{theorem}

\begin{proof}%
 Let $(G=(V,E),s,t)$ be a two-terminal series-parallel graph with edge lengths specified by~$\tau:E\to\N$.
 Let $(T,\lambda)$ be an sp-tree for $G$, where $\lambda$ is the labeling of the nodes of $T$.
 We identify each node~$\alpha\in V(T)$ with a two-terminal series-parallel graph $G_\alpha$ induced by the subtree rooted at $\alpha$.
 Recall that if $\rho\in V(T)$ is the root of $T$, then~$G_\rho=G$.

 Let $C[\alpha,x]$ denote the minimum number of edges to delete in $G_\alpha$ such that there is no path of length smaller than~$x$ connecting the two terminals.
 Observe that such an edge deletion set exists for every $x\in \N$, and its size is upper-bounded by the size of a minimum cut disconnecting the terminals.

 \emph{Case 1:} If $\alpha\in V(T)$ is a leaf of $T$ with $\lambda(\alpha)=\{v,w\}$, then
 \[
  C[\alpha,x] = \begin{cases}
                 1 ,& \text{if $\tau(\{v,w\})<x$} , \\
                 0 ,& \text{otherwise.}
                \end{cases}
 \]
 \emph{Correctness:}
  In the graph $G_\alpha=(\{v,w\},\{\{v,w\}\})$, we have to delete the edge~$\{v,w\}$ to increase the distance between $v$ and $w$ to $x$. This is possible if and only if $\tau(\{v,w\})<x$.

 \emph{Case 2:} If $\alpha\in V(T)$ is an inner node of~$T$ with $\lambda(\alpha)=\mathbf{S}$ and children~$\alpha_1$ and~$\alpha_2$, then
 \begin{align}
  C[\alpha,x] = \min_{x'\in\{0,\ldots,x\}} (C[\alpha_1,x']+C[\alpha_2,x-x']).\label{al:case2}
 \end{align}
 \emph{Correctness:}
 Let $G_\alpha$, $G_{\alpha_1}$, and $G_{\alpha_2}$ be the graphs corresponding to nodes $\alpha$, $\alpha_1$, and $\alpha_2$ respectively.
 Let $v,w$ denote the terminals of $G_\alpha$, and let $v',u'$ and $u'',w'$ be the terminals of $G_{\alpha_1}$ and $G_{\alpha_2}$ respectively.
 Recall that $G_\alpha$ is the serial composition of $G_{\alpha_1}$ and $G_{\alpha_2}$, thus $G_\alpha$ is obtained by identifying $u'$ with $u''$ as~$u$, and setting $v:=v'$ and $w:=w'$.

 Let $S\subseteq E(G_\alpha)$ be a set of $C[\alpha,x]$ edges such that there is no $vw$-path of length smaller than~$x$ in $G_\alpha-S$.
 Since $G_\alpha$ is the serial composition, $S=S_1\cup S_2$ with $S_1\subseteq E(G_{\alpha_1})$ and $S_2\subseteq E(G_{\alpha_2})$.
 Then there is $x^*\in\{0,\ldots,x\}$ with $\dist_{G_{\alpha_1}-S_1}(v',u')\geq x^*$ and $\dist_{G_{\alpha_2}-S_2}(u'',w')\geq x- x^*$ since every $vw$-path contains~$u$.
 It follows that
 \[ C[\alpha,x] = |S| = |S_1| + |S_2| \geq \min_{x'\in\{0,\ldots,x\}} (C[\alpha_1,x']+C[\alpha_2,x-x']). \]

 Conversely, let $x^*\in \{0,\ldots,x\}$ be such that the expression in \Cref{al:case2} is minimum.
 Let $S_1\subseteq E(G_{\alpha_1})$ and $S_2\subseteq E(G_{\alpha_2})$ with $|S_1|=C[\alpha_1,x^*]$ and $|S_2|=C[\alpha_2,x-x^*]$ such that there is no $v'u'$-path of length smaller than $x^*$ in $G_{\alpha_1}-S_1$ and no $u''w'$-path of length smaller than $x-x^*$ in $G_{\alpha_2}-S_2$.
 Let $S:=S_1\cup S_2$.
 Since every $vw$-path in $G$ contains the vertex $u$, it follows that $\dist_{G-S}(v,w) = \dist_{G-S}(v,u) + \dist_{G-S}(u,w)\geq x^* + x-x^*=x$.
 It follows that
 \[ \min_{x'\in\{0,\ldots,x\}} (C[\alpha_1,x']+C[\alpha_2,x-x']) = |S_1| + |S_2| = |S| \geq C[\alpha,x]. \]

 \emph{Case 3:} If $\alpha\in V(T)$ is an inner node of $T$ with $\lambda(\alpha)=\mathbf{P}$, and children~$\alpha_1$ and~$\alpha_2$, then
 \[
  C[\alpha,x] = C[\alpha_1,x]+C[\alpha_2,x].
 \]
\emph{Correctness:}
 Let $G_\alpha$, $G_{\alpha_1}$, and $G_{\alpha_2}$ be the graphs corresponding to nodes $\alpha$, $\alpha_1$, and $\alpha_2$, respectively.
 Let $v,w$ denote the terminals of $G_\alpha$, and let $v',w'$ and $v'',w''$ be the terminals of $G_{\alpha_1}$ and $G_{\alpha_2}$ respectively.
 Recall that $G_\alpha$ is the parallel composition of $G_{\alpha_1}$ and $G_{\alpha_2}$, thus $G_\alpha$ is obtained by identifying $v'$ with $v''$ as $v$ and $w'$ with $w''$ as $w$.

 Let $S\subseteq E(G_\alpha)$ be a set of $C[\alpha,x]$ edges such that there is no $vw$-path of length smaller than~$x$ in $G_\alpha-S$.
 Since $G_\alpha$ is the parallel composition, it holds that $S=S_1\cup S_2$ with $S_1\subseteq E(G_{\alpha_1})$ and $S_2\subseteq E(G_{\alpha_2})$. %
 Observe that there is a $vw$-path of length smaller than~$x$ in~$G-S$ if and only if there is a $v'w'$-path or a $v''w''$-path of length smaller than~$x$ in~$G_{\alpha_1}-S_1$ or in $G_{\alpha_2}-S_2$.
 The observation follows immediately from the definition of parallel compositions and the fact that $v'$ is identified with $v''$ as $v$ and $w'$ is identified with $w''$ as $w$.
 It follows that

 \[ C[\alpha,x] = |S| = |S_1| + |S_2| \geq C[\alpha_1,x]+C[\alpha_2,x]. \]

 Conversely, let $S_1\subseteq E(G_{\alpha_1})$ and $S_2\subseteq E(G_{\alpha_2})$ with $|S_1|=C[\alpha_1,x]$ and $|S_2|=C[\alpha_2,x]$ such that there is no $v'w'$-path of length smaller than~$x$ in~$G_{\alpha_1}-S_1$ and no $v''w''$-path of length smaller than $x$ in $G_{\alpha_2}-S_2$.
 Let $S:=S_1 \cup S_2$.
 Following the preceding observation, we obtain
 \[ C[\alpha_1,x]+C[\alpha_2,x] = |S_1|+|S_2| = |S| \geq C[\alpha,x]. \]

 We consider $C$ as a table in the remainder.
 We fill~$C$ in post-order on $T$, that is, whenever the entries for an inner node are to be filled, the entries of the child nodes are filled before.
 By the correctness of the cases above, if $\rho\in V(T)$ denotes the root of $T$, then $C[\rho,\ell]$ denotes the minimum number of edge deletions such that there is no $st$-path in $G$ of length smaller than~$\ell$.

 Since every edge in~$G$ one-to-one corresponds to a leaf in~$T$, there are $O(m)$~nodes in~$T$.
 Hence, the table $C$ has $O(m\cdot\ell)$ entries.
 In Case 2, we have to find a minimum in $O(\ell)$ time.
 Altogether, the algorithm takes $O(m\cdot\ell^2)$~time.
\end{proof}

\begin{remark}
	With a similar dynamic programming approach one can show an algorithm solving \kmve in~$O(m \cdot k^2)$ time, see \citet[Theorem 8.4]{Stahlberg16} for details.
	Furthermore, both the~$O(m \cdot \ell^2)$-time algorithm above and the~$O(m \cdot k^2)$-time algorithm extend to the case where the edges have integral edge-deletion costs.
	This problem variant with both edge-deletion costs and edge lengths was shown to be (weakly) NP-hard on series-parallel graphs with~$s$ and~$t$ being the two terminals by \citet{BaierEHKKPSS10}.
	The two algorithms above complement this with fixed-parameter tractability with respect to each~$k$ and~$\ell$.
\end{remark}

\medskip

In \Cref{thm!split} we showed that \SPMVE with unit-length edges on split graphs remains NP-hard.
Since split graphs are of diameter at most three, \SPMVE with unit-length edges remains NP-hard on graphs of diameter at least three.
The last result of this section shows that this bound on the diameter is strict. %

\begin{proposition}\label{prop:diamtwo}
	\SPMVE with unit-length edges is linear-time solvable on graphs of diameter at most two.
\end{proposition}

\begin{proof}%
	\citet{ItaiPS82} proved that for $\ell\leq 4$, SP-MVE with unit-length edges is solvable in polynomial time.
	Hence, it remains to consider the case where $\ell\geq 5$.
 
	\citet{PaynePVW12} showed that in any graph $H$ of diameter two, for each pair of distinct vertices $v,w\in V(H)$, there are $\min\{\deg(v),\deg(w)\}$ many edge-disjoint paths of length at most four.
	Hence, to achieve a distance of five or more between~$s$ and~$t$ we have to delete~$\min\{\deg(s),\deg(t)\}$ edges, which is sufficient to cut~$s$ from~$t$.
	Thus, any instance $(G,s,t,k,\ell)$ with $\ell\geq 5$ and $G$ being a graph of diameter two is a yes-instance if and only if $k\geq \min\{\deg(s),\deg(t)\}$.
	This can be decided in linear time.
\end{proof}

Observe that each connected component of a \emph{cograph} (a graph without an induced~$P_4$) has diameter two.
Note that threshold graphs are cographs.
Thus, the preceding result also shows that \SPMVE with unit-length edges is linear-time solvable on cographs and threshold graphs.

\section{Algorithms for some NP-hard cases}\label{sec_fpt}

In this section, we present fixed-parameter and approximation algorithms. %
First, we consider bounded-degree graphs.
Here, the basic observation is that the maximum vertex degree~$\Delta$ of a graph upper-bounds the number of deleted edges for SP-MVE: a budget of~$\Delta$ would allow to disconnect $s$ from~$t$ by deleting all edges incident to~$s$.

\begin{proposition}\label{prop:xp-delta}
	\SPMVE can be solved in~$O(m^{\Delta-1}\allowbreak (m+ n \log n))$ time.
\end{proposition}
{\begin{proof}
	Recall that we assume~$k$ to be smaller than the 
maximum degree~$\Delta$
as otherwise we could simply delete all edges incident to~$s$.
	The straightforward algorithm branching into all~$O(m^k)$ cases to delete at most~$k$ edges and checking with Dijkstra's shortest path algorithm whether the distance between~$s$ and~$t$ is high enough runs in~$O(m^k (m + n \log n)) = O(m^{\Delta-1} (m + n \log n))$ time.
\end{proof}}
The question whether one can replace~$m^{\Delta-1}$ by~$f(\Delta) \cdot m^{O(1)}$ for some function~$f$%
, that is, whether SP-MVE is not only in~XP  (as shown by Proposition~\ref{prop:xp-delta}) but also fixed-parameter tractable with respect to~$\Delta$, remains open.

\citet{GolovachT11} used a search tree algorithm to show that SP-MVE is fixed-parameter tractable when combining the parameters number~$k$ of removed edges and minimum $st$-path length~$\ell$ to be achieved.
We next state the result and describe the search tree since we will adapt it in the following.

\begin{proposition}[\citet{GolovachT11}]\label{prop:search-tree}
	\SPMVE can be solved in~$O((\ell-1)^k \cdot (n\log n + m))$ time.
\end{proposition}
\begin{proof}
	We employ a simple depth-bounded search tree:
	the basic idea is to search for a shortest $st$-path and to ``destroy'' it by deleting one of the edges (trying all possibilities).
	This is repeated until every shortest $st$-path has length at least~$\ell$.
	For each such shortest path, we branch into  at most~$\ell-1$ possibilities to delete one of its edges, and the depth of the corresponding search tree is at most~$k$ (our ``deletion budget'') since otherwise we cannot find a solution with at most~$k$ edge deletions.
	The correctness is obvious.
	Hence, we arrive at a search tree of size at most $(\ell-1)^k$ where in each step we need to compute a shortest path.
	Using Dijkstra's shortest algorithm, this can be done in~$O(n \log n + m)$ time.
	The overall running time is thus~$O((\ell-1)^k \cdot (n\log n + m))$.
\end{proof}
Using the search tree described in the proof of~\cref{prop:search-tree} to destroy all paths of length at most~$2^{O(\sqrt{\log n})}$ yields the following.

\begin{corollary}\label{cor:param-approx-k}
	For any constant~$c$, \kmve with unit-length edges can be approximated within a factor of $n/2^{c \cdot \sqrt{\log n}}$ in~$O(2^{k^2}k (n \log n + m) + n^{c^2 + 3})$ time.
\end{corollary}
{\begin{proof}
We employ the search tree algorithm behind
\cref{prop:search-tree};  it has size~$O((\ell-1)^k)$.
The idea now is
to either compute an optimal solution in fpt-time or to derive
the stated approximation in polynomial time.

Our parameterized approximation algorithm works as follows.
Trying $\ell =1, 2, \ldots , g(n)$ (where $g(n)$ is determined below) we employ the search tree
to detect whether there is an optimal solution of length
smaller than~$g(n)$. Namely, if the search tree for some
$\ell$-value says no, then we know that we found an optimal
solution with the previous search tree and output this.
Otherwise, we reach $\ell = g(n)$ and thus,
since the optimal value is
at most~$n-1$, this means that we have a factor-$n/g(n)$-approximation.

Overall, this procedure has at most $g(n)$ iterations and each has a running time of $O( g(n)^k \cdot (n \log n + m) )$.
It remains to determine for which (maximum) function~$g(n)$ this still yields fpt running time for parameter~$k$.
First, if $k>\log (g(n))$, then $g(n)^k = 2 ^ {k \cdot \log g(n)}$ can be upper-bounded by~$2^{k^2}$ and we are done.
Second, if $k\leq\log (g(n))$, then 
we have that $g(n)^k \leq g(n)^{\log (g(n))} = 2^{(\log (g(n))^2}$.
The latter term is polynomial if and only if $g(n) = 2^{O(\sqrt{\log n})}$.
More precisely, if for any constant~$c$ we have~$g(n) = 2^{c \cdot \sqrt{\log n}}$, then we get the bound~$2^{(\log (g(n))^2} \le n^{c^2}$.
In total the running time in this second case is bounded by~$O(n^{c^2 + 3})$.
\end{proof}}
By deleting every edge on too short $st$-paths, we obtain an~$\ell$-approximation.

\begin{proposition}\label{th:lapprox}
	\meb can be approximated within a factor of~$\ell$ in~$O(n^2 \log n + nm)$ time.
\end{proposition}
{\begin{proof}
	Let $I=(G=(V,E),\ell,s,t,\tau)$ be an instance of \meb.
	We repeat the following algorithm until the shortest $st$-path has length at least~$\ell$.
	Set $G':=G$ and let $P$ be a shortest $st$-path in $G'$.
	If the length~$\tau(P)$ of~$P$ is less than~$\ell$, then set $G' := G'-E(P)$ and proceed.
	Denote by $i$ the number of iterations the algorithm realizes.
	Let $E''$ be the set of all edges of the $i$ shortest paths removed from~$G$.
	The size of $E''$ is $|E''|\leq i\ell$ since at each step at most $\ell$ edges are deleted.
	Moreover, $\opt(I)\geq i$ since an optimal solution contains at least one edge of each of these $i$ paths.
	The number of iteration is at most~$n$ and each iteration can be done in~$O(n \log n + m)$ time.
\end{proof}}

\citet[Corollary 3.14]{BaierEHKKPSS10} provided a $b$-approximation algorithm for \meb running in~$O(b \cdot n \cdot m)$ time.
Observe that our approximation algorithm in \Cref{th:lapprox} provides a weaker approximation factor but a faster running time.

Combining the previous approximation algorithm  with a tradeoff between running time and approximation factor~\cite[Lemma 2]{BazganCNS14}, we obtain the following.
\begin{corollary}
	For every increasing function~$r$, \meb is parameterized $r(n)$-approximable with respect to the parameter~$\ell$.
\end{corollary}

\paragraph{Parameter feedback edge set number.}
We next provide a linear-size problem kernel for SP-MVE parameterized by the feedback edge set number.
An edge set~$F \subseteq E$ is called \emph{feedback edge set} for a graph~$G = (V,E)$ if~$G-F$ is a tree or a forest.
The feedback edge set number of~$G$ is the size of a minimum feedback edge set.
Note that if~$G$ is connected, then the feedback edge set number equals~$m-n+1$.
Computing a spanning tree, one can determine a minimum feedback edge set in linear time.
Hence, we assume in the following that we are given a feedback edge set~$F$ with~$|F|=f$ for our input instance~$(G=(V,E),k,\ell,s,t,\tau)$.
We start with two simple data reduction rules dealing with degree-one and degree-two vertices.

\begin{rrule}\label{rule:degreeOneVertices}
	Let~$(G=(V,E),k,\ell,s,t,\tau)$ be an SP-MVE instance and let~$v \in V \setminus \{s,t\}$ be a vertex of degree one. Then, delete~$v$.
\end{rrule}
The correctness of \cref{rule:degreeOneVertices} is obvious as no shortest path uses a degree-one vertex.
We deal with degree-two vertices as follows.

\begin{rrule}\label{rule:degreeTwoVertices}
	Let~$(G=(V,E),k,\ell,s,t,\tau)$ be an SP-MVE instance and let~$v \in V \setminus \{s,t\}$ be a vertex of degree two with~$N_G(v) = \{u,w\}$ and~$\{u,w\} \notin E$. %
	Then add the edge~$\{u,w\}$ with the length~$\tau(\{u,w\}) := \tau(\{u,v\})+\tau(\{v,w\})$ and delete~$v$.
\end{rrule}
The correctness of \cref{rule:degreeTwoVertices} follows from the fact that on an induced path at most one edge will be deleted and it does not matter which one will get deleted.
Applying both rules exhaustively can be done in linear time and leads to the following problem kernel.

\begin{theorem} \label{thm:fes-linear-kernel}
	\SPMVE admits a linear-time computable problem kernel with~$5f+2$ vertices and~$6f+2$ edges.
\end{theorem}
{\begin{proof}
	Let~$(G=(V,E),k,\ell,s,t,\tau)$ be the input instance of SP-MVE.
	First, we exhaustively apply \cref{rule:degreeOneVertices,rule:degreeTwoVertices}.
	It remains to upper-bound the size of the reduced graph~$G'$.
	To this end, first observe that~$G'$ contains at most~$f$ degree-two vertices as every degree-two vertex that is not deleted by \cref{rule:degreeTwoVertices} has two neighbors that are adjacent to each other and thus induces together with its neighbors a cycle.
	It remains to upper-bound the number of vertices with degree at least three.
	To this end, let~$r$ denote the number of leaves in the tree~$G' - F$.
	Thus, $G'-F$ contains at most~$r-2$ vertices of degree at least three.
	Due to \cref{rule:degreeOneVertices}, $G'$ contains at most two degree-one vertices ($s$ and~$t$) and, hence, $r \le 2f+2$.
	Furthermore, there are at most~$2f$ degree-three vertices in~$G'$ that are incident to an edge in~$F$.
	Hence, $G'$ contains at most~$4f+2$ vertices of degree at least three.
	In total, $G'$ contains at most~$5f+2$ vertices and, thus, $6f+2$ edges. 
	
	We now discuss the running time.
	To apply the rules, start with sorting the vertices by degree in non-decreasing order.
	Since all degrees are smaller than~$n$, the sorting can be done in~$O(n)$ time using e.\,g. Bucket sort.
	Then, deleting all degree-one vertices and updating their neighbors' degrees can be done in linear time.
	Similarly, once \cref{rule:degreeOneVertices} is no more applicable, the degree-two vertices can be dealt with in similar fashion. 
	Note that applying \cref{rule:degreeTwoVertices} does not change the degrees of the neighbors of the degree-two vertex. 
	Thus, for each degree-two vertex removing it and adding the extra edge can be done in constant time.
	Hence, the overall time to apply both rules is linear.
\end{proof}}

By simply trying all possibilities to delete edges in the problem kernel and checking with Dijkstra's algorithm the distance between~$s$ and~$t$, we obtain the following.

\begin{corollary}
	\SPMVE can be solved in~$O(2^{6f} (n\log n + m))$ time where~$f$ is the feedback edge set number.
\end{corollary}

\paragraph{Parameter cluster vertex deletion number.}
We now prove that SP-MVE restricted to unit-length edges is fixed-parameter tractable with respect to the parameter cluster vertex deletion number~$x$.
A graph $G$ is a \emph{cluster graph} if it is a disjoint union of cliques.
A vertex set~$X \subseteq V$ is called \emph{cluster vertex deletion set} if~$G[V \setminus X]$ is a cluster graph~\cite{HKMN10}.
The cluster vertex deletion number is the size of a minimum cluster vertex deletion set.

Recall that \SPMVE with arbitrary edge lengths is NP-complete on complete graphs (see \cref{thm:cliques-np-hard}).
Thus, the algorithm presented below for the unit-length case cannot be extended to the more general case with arbitrary edge lengths since a clique has cluster vertex deletion number zero.

We assume in the following that for the input instance~$(G = (V,E),k,\ell,s,t)$ we are given a cluster vertex deletion set~$X$ of size~$x$.
If~$X$ is not already given, then we can compute~$X$ in~$O(1.92^x \cdot (n+m))$ time~\cite{BCKP16}.
Our algorithm is based on the observation that twins can be handled equally in a solution.
This follows from a more general statement provided in the following lemma.
It shows that for any set~$T \subseteq V \setminus \{s,t\}$ of vertices that have the same neighborhood in~$V \setminus T$, we can assume that we do not delete edges in~$G[T]$ and that the vertices in~$T$ behave the same, that is, one deletes either all edges or no edge between a vertex~$v \in V \setminus T$ and the vertices in~$T$.

\begin{lemma}\label{lem:twinUnitLentghs}
	Let $G = (V,E)$ be an undirected graph with unit-length edges, let $s,t \in V$ be two vertices, and let~$T = \{v_1,\ldots,v_t\} \subseteq V \setminus \{s,t\}$ be a set of vertices such that~$N_G(v_1) \setminus T = N_G(v_2) \setminus T = \ldots = N_G(v_t) \setminus T$.
	Then, for every edge subset~$S \subseteq E$, there exists an edge subset~$S' \subseteq E$ such that~$\dist_{G-S'}(s,t) \ge \dist_{G-S}(s,t)$, $|S'| \le |S|$, and~$N_{G[S']}(v_1) = N_{G[S']}(v_2) = \ldots = N_{G[S']}(v_t)$.
\end{lemma}
{\begin{proof}
	Starting from the edge subset~$S \subseteq E$, we construct~$S'$ having the desired properties.
	To this end, we abbreviate~$\ell := \dist_{G-S}(s,t)$.
	Let~$T \in V \setminus \{s,t\}$ be a set of vertices such that~$N(u) \setminus T = N(v) \setminus T$ for each pair~$u,v \in T$.
	Assume that the vertices in~$T$ do not have the same neighborhood in~$G[S]$; otherwise, we simply set~$S' := S$.
	Let~$u \in T$ be a vertex that has in the graph~$(V,S)$ the smallest degree of all vertices in~$T$, that is, the vertex in~$T$ that is incident to the least number of edges in~$S$.
	Now, construct~$S'$ as follows.
	First, initialize~$S'$ as a copy of~$S$.
	Second, remove all edges of~$S'$ that have both endpoints in~$T$.
	Third, for each~$v \in T \setminus\{u\}$ remove all edges incident to~$v$ from~$S$ and add for each edge~$\{u,w\} \in S$ the edge~$\{v,w\}$.
	Summarizing, $S'$ is composed as follows:
	\begin{align*}
		S' := {} & (S \setminus \{\{v,w\} \mid v \in T\setminus\{u\} \wedge w \in V \})\ \cup \\ & \{\{v,w\} \mid v \in T \setminus \{u\} \wedge w \in V \setminus T \wedge \{u,w\} \in S\}.
	\end{align*}
	By construction of~$S'$ we have~$|S'| \le |S|$.
	Furthermore, we have~$N_{G[S']}(v) = N_{G[S]}(u) \setminus T$ for all~$v \in T$ and thus $N_{G[S']}(v) = N_{G[S']}(v')$ for each pair~$v,v' \in T$.
	It remains to show that in~$G - S'$ the distance between~$s$ and~$t$ is at least~$\ell$.
	To this end, assume by contradiction that $G-S'$ contains an $st$-path~$P$ of length less than~$\ell$.
	Since, by construction of~$S'$, each edge in~$S \setminus S'$ has at least one endpoint in~$T$, it follows that~$P$ contains at least one vertex of~$T$.
	Let~$v$ and $v'$ be the first respectively last vertex of~$T$ on~$P$ (possibly~$v = v'$) and let~$w,w'$ be the vertices before~$v$ respectively after~$v'$ on~$P$, that is $$P = s\text{-}\ldots\text{-}w\text{-}v\text{-}\ldots\text{-}v'\text{-}w'\text{-}\ldots\text{-}t.$$
	Since~$w,w' \notin T$, $N_G(v) \setminus T = N_G(v') \setminus T$, and~$N_{G[S']}(v) = N_{G[S']}(v')$, it follows that~$Pw$-$v$-$w'P$ is also an $st$-path with length less than~$\ell$ in~$G-S'$.
	Similarly, it follows that~$P' := Pw$-$u$-$w'P$ is also an $st$-path with length less than~$\ell$ in~$G-S'$ (where~$u$ is the vertex used in the construction of~$S'$).
	Since~$N_{G[S']}(u) = N_{G[S]}(u) \setminus T$ it follows that~$\{u,w\},\{u,w'\} \notin S$, implying that~$P'$ is an~$st$-path of length less than~$\ell$ in~$G-S$; a contradiction to the assumption that~$\dist_{G-S}(s,t)=\ell$.
\end{proof}}

Using \cref{lem:twinUnitLentghs} we can show that \SPMVE with unit-length edges is linear-time fixed-parameter tractable with respect to the parameter cluster vertex deletion number.

\begin{theorem}\label{thm:fptclusterdel}
	\SPMVE with unit-length edges can be solved in~$2^{2^{O(x)}}(n+m)$ time where~$x$ is the cluster vertex deletion number.
\end{theorem}

\begin{proof}
	Let~$(G = (V,E),k,\ell,s,t)$ be the input instance of SP-MVE and let~$X \subseteq V$ be a cluster vertex deletion set of size~$x$.
	Hence, $G-X$ is a cluster graph and the vertex sets $C_1, \ldots, C_r$ form the cliques (clusters) for some~$r \in \N$.
	We set~$\mathcal{C} := \{C_1, \ldots, C_r\}$.
	Assume that there is an SP-MVE solution~$S \subseteq E$ of size at most~$k$; otherwise the algorithm will output `no' as it finds no solution.
	We describe an algorithm that finds~$S$.

	Our algorithm is based on the following observation.
	Let~$P$ be an arbitrary shortest $st$-path that goes through a clique~$C \in \mathcal{C}$ in~$G-S$.
	Then, $P$ contains at most~$2^x$ vertices from~$C$:
	By \cref{lem:twinUnitLentghs}, we can assume that the twins in~$G$ are still twins in~$G-S$.
	Since~$P$ is a shortest path, $P$ does not contain two vertices that are twins.
	As the vertices in~$C$ form a clique, they only differ in how they are connected to vertices in~$X$.
	Thus, $C$ contains at most~$2^x$ ``different'' vertices, that is, vertices with pairwise different neighborhoods.
	
	Now, consider two non-adjacent vertices~$u,v \in X$.
	From the above considerations it follows that in~$G-S$ a $uv$-path avoiding the vertices in~$X$ has length between one and~$2^x+1$ as it can pass through at most one clique.
	Our algorithm tries for each vertex pair from~$X$ all possibilities for the distance it has in~$G-S$ and then tries to realize the current possibility.
	After the current possibility is realized, the cliques in~$\mathcal{C}$ are obsolete and thus the instance size can be upper-bounded in a function of~$x$.
	More precisely, our algorithm works as follows:
	\begin{enumerate}
		\item Branch into all possibilities to delete edges contained in~$G[X]$. Decrease the budget~$k$ accordingly.\label{step:edgesInCVDSet}
		\item Branch into all possibilities to add for each pair~$u,v$ of non-adjacent vertices in~$X$ an edge with a length lying in~$\{2,3,\ldots,2^x,2^x+1,\infty\}$ and indicating the length of a shortest path between~$u$ and~$v$ that does not contain any vertex in~$X$.\label{step:assignDistances}
		\item Delete for each clique containing neither~$s$ nor~$t$ the \emph{minimum number} of edges to ensure that a shortest path between each pair of vertices in~$X$ is completely contained in~$G[X]$. Decrease the budget~$k$ accordingly.\label{step:ensureDistances}
		\item Remove all cliques except the ones that contain $s$ or $t$. Do \emph{not} change the budget~$k$.\label{step:remove cliques}
		\item Solve the problem on the remaining graph with the remaining budget (that was not spent in \cref{step:edgesInCVDSet,step:ensureDistances}).\label{step:solveRemainingProblem}
	\end{enumerate}
	Note that \cref{step:assignDistances} is performed for each possibility in \cref{step:edgesInCVDSet}.
	Hence, in \cref{step:edgesInCVDSet,step:assignDistances} at most $2^{x^2} \cdot (2^x+1)^{x^2}$ possibilities are considered and for each of these possibilities \cref{step:ensureDistances} is invoked.

	In \cref{step:ensureDistances}, the algorithm tries to realize the prediction made in \cref{step:assignDistances}.
	To this end, let~$C \in \mathcal{C}$ be a clique containing neither~$s$ nor~$t$.
	The algorithm branches into all possibilities to delete edges in~$G[C]$ or edges with one endpoint in~$C$ and the other endpoint in~$X$.
	Since~$G[C]$ contains at most~$2^x$ different vertices, it follows from \cref{lem:twinUnitLentghs} that at most~$2^{(2^x)^2+2^x \cdot x} = 2^{(4^x)+2^x \cdot x}$ possibilities need to be considered to delete edges.
	For each possibility, the algorithm checks in $x^{O(1)}$ time whether all shortest paths between a pair of vertices of~$X$ go through~$C$.
	If yes, then the algorithm discards the currently considered branch; if no, then the current branch is called valid.
	From all valid branches for~$C$, the algorithm picks the one that deletes the minimum amount of edges and proceeds with the next clique.
	Observe that since~$X$ is a vertex separator for all cliques in~$\mathcal{C}$, the algorithm can solve \cref{step:ensureDistances} for each clique independently of the outcome in the other cliques.
	Hence, the overall running time for \cref{step:ensureDistances} is~$2^{2^{O(x)}} \cdot n$ as~$|\mathcal{C}| \le n$.
	
	As discussed above, the cliques in~$\mathcal{C}$ containing neither~$s$ nor~$t$ are now obsolete as there is always a shortest path avoiding these cliques.
	Hence, the algorithm removes these cliques (\cref{step:remove cliques}).
	This can be done in linear time.
	The remaining instance consists of the vertices in~$X$ and the at most two cliques containing~$s$ and~$t$.
	As the algorithm deleted the edges within~$G[X]$ in \cref{step:edgesInCVDSet}, it remains to consider deleting edges within the two cliques or between the two cliques and the vertices in~$X$.
	Again, by \cref{lem:twinUnitLentghs}, the algorithm only needs to branch into~$2^{2 \cdot (4^x + x \cdot 2^x)}$ possibilities to delete edges and check for each branch whether~$s$ and~$t$ have distance at least~$\ell$ and the overall budget~$k$ is not exceeded.
	If one branch succeeds, then the algorithm found a solution and returns it.
	If no branch succeeds, then there exists no solution of size~$k$ since the algorithm performed an exhaustive search.
	Overall, the running time is~$2^{2^{O(x)}} \cdot (n+m)$.
\end{proof}
Obviously, it would be interesting to improve the above algorithm 
by obtaining linear-time fixed-parameter tractability with a 
single-exponential-time algorithm.

\section{Conclusion}\label{sec_conclusion}
The \MVE (\SPMVE) problem is a natural edge deletion problem that 
is amenable to a rich body of fine-grained (multivariate) 
computational complexity 
analysis. Such a study has been initiated here, identifying numerous 
challenges for future work. 
\cref{fig:par-hier} in the introductory section depicts a wide range of graph parameters for which the parameterized complexity status of \SPMVE is unknown.
Also concerning the approximation point of view not much is known.
There is a huge gap between the known lower and upper bounds of the approximation factor achievable in polynomial time.
Further, from a practical point of view it would make sense to extend our studies by restricting the input to planar graphs~\cite{PS16,FHNN18}---here
one might hope for further fixed-parameter tractability results.
Moreover, the complexity of \SPMVE remains open even for highly structured graphs such as interval or proper interval graphs; we conjecture that \SPMVE is polynomial-time solvable on proper interval graphs~\cite{Stahlberg16}.
Finally, also in terms of parameterized approximability~\cite{Marx08} \MVE offers a number of interesting challenges for future work.
	\bibliographystyle{abbrvnat}
	\bibliography{mve_arxiv}

\end{document}